\documentclass[a4paper, 12pt, headings = small, abstract]{scrartcl}
\usepackage[english]{babel}

\usepackage{amscd,amsmath,amsthm,amssymb,array,hhline,mathrsfs, enumerate, booktabs,fancybox,calc,xcolor,graphicx,bbm,xspace,nicefrac,stmaryrd,url,arcs,stmaryrd,pifont}

\usepackage{authblk}

\usepackage[sfmathbb,partialup]{kpfonts}

\DeclareMathAlphabet{\mathsf}{OT1}{\sfdefault}{m}{n}
\SetMathAlphabet{\mathsf}{bold}{OT1}{\sfdefault}{b}{n}
\usepackage[utf8]{inputenc}

\usepackage{enumitem}
\usepackage{etoolbox}
\usepackage[normalem]{ulem}
\usepackage{mathtools}
\usepackage{geometry}
\usepackage{float}

\usepackage{bm}
\usepackage{tikz}
\usepackage[outdir =./]{epstopdf}
\numberwithin{equation}{section}

\usepackage{xcolor}
\usepackage[colorlinks=true, allcolors=myteal]{hyperref}
\definecolor{grey_pers}{RGB}{69 90 100}
\definecolor{WIMgreen}{RGB}{60 134 132}
\definecolor{red_pers}{RGB}{204 37 41}
\definecolor{UMblue}{RGB}{4 47 86}
\definecolor{myteal}{RGB}{0 123 137}
\definecolor{dartmouthgreen}{rgb}{0.05, 0.5, 0.06}\definecolor{cobalt}{rgb}{0.0, 0.28, 0.67}\definecolor{coolblack}{rgb}{0.0, 0.18, 0.39}
\definecolor{glaucous}{rgb}{0.38, 0.51, 0.71}\definecolor{hooker\'sgreen}{rgb}{0.0, 0.44, 0.0}\definecolor{lemonchiffon}{rgb}{1.0, 0.98, 0.8}\definecolor{oucrimsonred}{rgb}{0.6, 0.0, 0.0}\definecolor{radicalred}{rgb}{1.0, 0.21, 0.37}\definecolor{raspberry}{rgb}{0.89, 0.04, 0.36}\definecolor{royalazure}{rgb}{0.0, 0.22, 0.66}
\definecolor{dex}{RGB}{138 18 34}

\theoremstyle{plain}
\newtheorem{theorem}{Theorem}

\newtheorem{lemma}[theorem]{Lemma}

\theoremstyle{definition}

\theoremstyle{assumption}

\theoremstyle{remark}

\setkomafont{sectioning}{\normalcolor\bfseries}
\setkomafont{descriptionlabel}{\normalcolor\bfseries}
\setkomafont{section}{\Large}
\setkomafont{subsection}{\large}
\setkomafont{subsubsection}{\large}
\setkomafont{paragraph}{\large}
\setkomafont{subparagraph}{\large}
\setkomafont{author}{\large}
\setkomafont{date}{\normalsize}

\def\D{\mathbb{D}}
\def\E{\mathbb{E}}

\def\F{\mathcal{F}}

\definecolor{darkred}{rgb}{0,0.6,0}
\def\X{\mathbf{X}}

\renewcommand{\hat}{\widehat}

\restylefloat{table}

\let\originalleft\left
\let\originalright\right
\renewcommand{\left}{\mathopen{}\mathclose\bgroup\originalleft}
\renewcommand{\right}{\aftergroup\egroup\originalright}

\renewcommand{\X}{\mathbf{X}}
\newcommand{\Y}{Y}
\newcommand{\hY}{\hat{Y}}
\renewcommand{\th}{\theta}
\newcommand{\thf}{\boldsymbol{\theta}}
\newcommand{\Zf}{\mathbf{Z}}
\newcommand{\Z}{Z}
\renewcommand{\u}{U}
\newcommand{\uf}{\mathbf{U}}
\newcommand{\n}{\prescript{n}{}}
\renewcommand{\D}{D}
\newcommand{\Df}{\mathbf{D}}

\newcommand{\Vf}{\mathbf{V}}
\newcommand{\tn}{\left \lfloor{tn}\right \rfloor/n}

\newcommand{\Cf}{\mathbf{C}}

\newcommand{\Mf}{\mathbf{M}}

\newcommand{\delf}{\boldsymbol{\delta}}

\allowdisplaybreaks

\makeatother
\title{\fontsize{16}{19} \selectfont Is Learning in Biological Neural Networks based on Stochastic Gradient Descent?}
\subtitle{An analysis using stochastic processes}

 \author[1]{Sören Christensen}
 \author[1]{Jan Kallsen}
 \affil[1]{Department of Mathematics, Kiel University}
\date{\today}

\begin{document}
\maketitle
{
\begin{abstract}
In recent years, there has been an intense debate about how learning in biological neural networks (BNNs) differs from learning in artificial neural networks. It is often argued that the updating of connections in the brain relies only on local information, and therefore a stochastic gradient-descent type optimization method cannot be used. In this paper, we study a stochastic model for supervised learning in BNNs. We show that a (continuous) gradient step occurs approximately when each learning opportunity is processed by many local updates. This result suggests that stochastic gradient descent may indeed play a role in optimizing BNNs.
\end{abstract}


\vspace{.2cm}

\section{Introduction}
In order to understand how biological neural networks (BNNs) work, it seems natural to compare them with artificial neural networks (ANNs). Although the definition of the latter is inspired by the former, they also differ in several aspects.  One of them is the way the network parameters are updated. 

In simple terms, an ANN learns from data by adjusting the weights of the connections between nodes in order to minimize a loss function that measures the difference between the desired output and the actual output of the network. More specifically, the optimization step is performed using the Stochastic Gradient Descent (SGD) algorithm, which iteratively updates the weights of the network by moving them in the direction of the steepest descent of the empirical loss function of a single training sample.  The gradient itself is computed with the so-called backpropagation algorithm. In particular, the update of any parameter is based on the states of all other parameters. Such a mechanism does not seem to be biologically plausible for BNNs, as many authors have pointed out. Parameter update in BNNs occurs only locally, and distant neurons are only indirectly connected through the endogenous reward system. This observation is closely related to the weight transportation problem \cite{lillicrap2020backpropagation,crick1989recent,grossberg1987competitive}. We refer to \cite{whittington2019theories,tavanaei2019deep} for a detailed discussion about the role of SGD in BNN, which the author of \cite[Section 5]{schmidt2023interpreting} summarizes as follows: 
``[T]here are various theories that are centered around the idea that the learning in BNNs should be linked to gradient descent. All of these approaches, however, contain still biological implausibilities and lack a theoretical analysis.''

The starting point for the present paper is the recent article \cite{schmidt2023interpreting} just cited. In this seminal study, the author proposes a very persuasive stochastic model for brain-supervised learning
which has a thorough biological foundation in terms of spike-timing-dependent plasticity.
We review and discuss this setup in Section \ref{sec:model_SH}. In this model the local updating rule of the connection parameters in BNNs turns out to be a zero-order optimization procedure. More precisely, it is shown in \cite{schmidt2023interpreting} that the expected value of the iterates coincides with a modified gradient descent. However, this holds only on average. The noise for such zero-order methods is so high that one can hardly imagine effective learning based on it, see \cite{duchi2015optimal,nesterov2017random,conn2009introduction}. The author himself writes in  \cite[Section 4]{schmidt2023interpreting}: ``It remains to reconcile the observed efficiency of learning in biological neural networks with the slow convergence of zero-order methods.'' 

In this paper we make an attempt to achieve this reconciliation. To this end, we consider in Section \ref{sec:unser_model} a slight modification of the model of \cite{schmidt2023interpreting}. More specifically, we relax the assumption that for each learning  opportunity and each connection, exactly one spike is released. Instead, we assume that a large number of spikes is released for each training sample in the BNN model and thus many parameter updates are made {for any observed item}.
{Our revised model can be thought of as the original Schmidt-Hieber system which receives and processes each
input-output pair $n$ times in a row rather than once.}
It turns out that with this modification, the updates correspond approximately to a continuous descent step along the gradient flow, see Theorem~\ref{thm:convergence}. This can be interpreted in the sense that it is not biologically implausible that BNNs use a kind of SGD algorithm after all, but without explicitly computing the gradient.

\section{The Schmidt-Hieber model for BNNs revisited}\label{sec:model_SH}

We begin this section by reviewing the model introduced in \cite{schmidt2023interpreting}. It considers a classical instance of supervised learning: input-output pairs $(\X_1,\Y_1), (\X_2,\Y_2),\dots$ are given as observations, all being identically distributed. The goal is to predict the output $\Y$ for each new input $\X$ based on previous training data. This setting includes, for example, classification (when the set of possible outcomes of\ $\Y$ is finite) or regression problems. 

A (feedforward) biological neural network (BNN) is modeled in \cite{schmidt2023interpreting} as a directed acyclic graph with input neurons receiving information from the observations $\X_k$ and generating a predicted response $\hY_k$ as output. The nodes represent the neurons in the network and an edge $\nu=(i,j)$ between two nodes $i$ and $j$ indicates that neuron $i$ is presynaptic for neuron $j$. Each element  $(i,j)$ in the edge set $\mathcal T$ has a weight $w_{ij}$ which indicates the strength of the connection between $i$ and $j$. While the structure of the graph does not change, the weights $w_{ij}$ are adjusted in each learning step. 

Spike-timing-dependent plasticity (STDP) is chosen as the biological mechanism to update the parameters. It is considered as a form of Hebbian learning \cite{hebb2005organization}, which states that neurons that fire together wire together. More precisely, the synaptic weight $w_{ij}$ changes depending on the timing of the spikes from neuron $i$ to neuron $j$. The weight decreases when neuron $i$ spikes before neuron $j$, and increases when neuron $j$ spikes before neuron $i$. The closer the spikes are in time, the larger the change in weight. It is important to note that the spike times are modeled as random variables. After some standardization (see \cite[Equation (4.3)]{schmidt2023interpreting}) the update of the parameter for edge $(i,j)$ becomes 
\[w_{ij} \leftarrow w_{ij} + w_{ij}C(e^{-\u_{ij}} - e^{\u_{ij}}),\]
where $\u_{ij}$ are uniformly distributed random variables on some interval $[-A,A]$, i.e.\ $\u_{ij}\sim \mathcal U(-A,A)$, modeling the random spike times. 
The constant $C$ represents the effect of a reward system, for example by neurotransmitters such as dopamine. It plays a key role for
any meaningful learning process.
In the present setup, the reward is tied to the success of predicting $\Y_k$, more specifically, whether the task is solved better or worse than in earlier trials. Using the restandardization $\theta_{ij}:=\log w_{ij}$ and a Taylor approximation, the following structure is derived in \cite[Equation (4.5)]{schmidt2023interpreting} for the update of $\th_{ij}$ at step $\ell$:
\begin{align}\label{eq:general_update}
\th_{ij}^{(\ell)} = \th_{ij}^{(\ell-1)} + \alpha^{(\ell-1)} \left(L^{(\ell-1)}(\thf^{(\ell-1)}+\uf^{(\ell)})-L^{(\ell-2)}(\thf^{(\ell-2)}+\uf^{(\ell-1)})\right)\bigl(e^{-\u_{ij}^{(\ell)}} - e^{\u_{ij}^{(\ell)}}\bigr),
\end{align}
where $\alpha^{(\ell)}>0$ is a learning rate, $\thf^{(\ell)}=\big(\theta^{(\ell)}_{ij}\big)_{(i,j)\in\mathcal T}$ denotes the vector of parameters for all edges, 
$\uf^{(\ell)}$ the vector of the independent uniformly distributed random variables $U^{(\ell)}_{ij}$ for all edges $(i,j)$, 
and  $L^{(\ell)}$ the loss function associated with the respective step $\ell$. Thus, the update of the weights of the individual edges
is in fact affected by the state of the entire network, but only through the value of the common loss function, which provides an assessment of the learning success. In particular, no gradient appears in the mechanism. 

In \cite{schmidt2023interpreting} the author considers the case where for each input-output pair $(\X_k,\Y_k)$ the parameters are updated only once. To this end, the loss of the current input-output pair is compared with the previous one. More specifically, 
we have $\ell=k$ as well as $L^{(k)}(\thf)=L(\thf,\X_k,\Y_k)$, $k=1,2,\dots$, leading to the update rule
\begin{align}\label{eq:update_SH}
	\th_{ij}^{(k)} = \th_{ij}^{(k-1)} + \alpha^{(k-1)} \left(L(\thf^{(k-1)}+\uf^{(k)},\X_{k-1},\Y_{k-1})-L(\thf^{(k-2)}+\uf^{(k-1)},\X_{k-2},\Y_{k-2})\right)\bigl(e^{-\u_{ij}^{(k)}} - e^{\u_{ij}^{(k)}}\bigr).
\end{align}
As a main result, the author shows in \cite[Theorem 1]{schmidt2023interpreting} that this procedure 
corresponds on average to a gradient descent method, with a gradient evaluated not exactly at $\thf^{(k-1)}$ but slightly perturbed randomly.
However, as noted in \cite{schmidt2023interpreting}, sufficiently fast convergence cannot be expected for such a zero-order method. 

\section{Multiple updates per learning opportunity}\label{sec:unser_model}
The key ingredient leading to \eqref{eq:update_SH} is that, for each learning opportunity and each connection, exactly one spike is triggered and thus only one update of the parameters is made. The author himself calls this assumption ``strong'', see \cite[Section 4]{schmidt2023interpreting}. 

Given the average spike frequency in real biological systems and the strong brain activity even at immobile rest \cite{liu2021single}, it seems more reasonable to assume instead a large number of spikes per learning opportunity. 
This corresponds to a series $\thf^{(k,0)},\dots,\thf^{(k,n)}$ of updates to the parameters after observing any input-output pair $(\X_{k-1},\Y_{k-1})$. The assessment of the update steps is based on the loss function associated with the most recent observation $(\X_{k-1},\Y_{k-1})$. More specifically, equation \eqref{eq:general_update}  turns into
\begin{align}\label{eq:unser_update0}
\th_{ij}^{(k,\ell)} ={}& \th_{ij}^{(k,\ell-1)} + \alpha^{(k-1,\ell-1)} \left(L(\thf^{(k,\ell-1)}+\uf^{(k,\ell)},\X_{k-1},\Y_{k-1})-L(\thf^{(k,\ell-2)}+\uf^{(k,\ell-1)},\X_{k-1},\Y_{k-1})\right)\nonumber\\
&\times \bigl(e^{-\u_{ij}^{(k,\ell)}} - e^{\u_{ij}^{(k,\ell)}}\bigr)
\end{align}
for $\ell\geq 1$ with initial values given by $\thf^{(k,0)}:=\thf^{(k,-1)}:=\thf^{(k-1,n)}$. 
Since $n$ is considered to be large, the individual update steps should be small in order to avoid overfitting. 

We start by analyzing this update rule for a fixed observation $(\X_{k-1},\Y_{k-1})$, i.e.\ for a fixed $k$.
For ease of notation we suppress the dependence on $k$ in the following considerations. 
Since our goal is to study the limiting behavior for a large number of update steps $n$, 
we instead make the dependence on $n$ explicit. So we rewrite \eqref{eq:unser_update0} as
\begin{align}\label{eq:unser_update}
\n\th_{ij}^{(\ell)} = \n\th_{ij}^{(\ell-1)} + \n\alpha^{(\ell-1)} \left(L(\n\thf^{(\ell-1)}+\n\uf^{(\ell)})-L(\n\thf^{(\ell-2)}+\n\uf^{(\ell-1)})\right)\bigl(e^{-\n\u_{ij}^{(\ell)}} - e^{\n\u_{ij}^{(\ell)}}\bigr).
\end{align}
The parameter update depends on increments of a loss function which resembles a gradient on first glance. Note, however, that this increment term is the same for all edges and, moreover, randomness occurs in both the loss function and in the external factor.
We now consider the following dependencies on $n$:
\begin{align}\label{eq:rate}
	\n\alpha^{(\ell-1)}:=\alpha, \quad A_n:=n^{-1/3}A, \quad\n\u_{ij}^{(\ell)}\sim \mathcal U\left(-A_n,A_n\right),
\end{align}
with constants $\alpha, A>0$.
Moreover, we rescale and extend the discrete-time process $(\n\th_{ij}^{(\ell-1)})_{\ell=-1,\dots,n}$ in time, defining a 
continuous-time process  $\Zf^n=(\Zf^n_t)_{t\in[-1/n,1]}$  by
\begin{align}\label{eq:Z_def}
	\Zf^n_{t_\ell^n}=\Zf^n_{t_{\ell-1}^n} + \alpha_n \left(L(\Zf^n_{t_\ell^n}+\n\uf^{(\ell)})-L(\Zf^n_{t_{\ell-2}^n}+\n\uf^{(\ell-1)})\right)\left(e^{-\n\u_{ij}^{(\ell)}} - e^{\n\u_{ij}^{(\ell)}}\right).
\end{align}
for $t_\ell^n:={\ell}/{n}$
and by $\Zf^n_t:=\Zf^n_{\left \lfloor{tn}\right \rfloor/n}$ 
for arbitrary $t\in[0,1]$.

As a candidate limit for large $n$ we consider a standard rescaled gradient process $\Zf=(\Zf_t)_{t\in[0,1]}$, 
which is defined as the solution to the deterministic ordinary differential equation (ODE)
\begin{align}\label{eq:Z_stetig}
\frac{d\Zf_t}{dt}=-{{2\over3}A^2}\alpha\nabla L(\Zf_t),\quad \Zf_0=\thf^{(k,0)},
\end{align}
see e.g.\ \cite{orvieto2019shadowing}.
In order for our main theorem to hold, we assume that  
\begin{align}\label{ass:1}
\mbox{$\nabla L$ is bounded and Lipschitz continuous with Lipschitz constant $\lambda$}.
\end{align}

$\Zf$ naturally emerges as a limit if you run the ordinary gradient descent algorithm for minimizing the function $L$ with many small steps.
The main result of this paper states that the rescaled STDP process $\Zf^n$ converges to $\Zf$ as well. More precisely, we have
\begin{theorem}\label{thm:convergence}
Assume \eqref{eq:rate} and \eqref{ass:1}. Then, for each fixed training sample $k$, the rescaled process  $\Zf^n$ of the BNN weights  converges to the rescaled gradient process $\Zf$ uniformly in $L^2$, i.e.\ 
\[\lim_{n\to\infty}\E\left(\sup_{t\in[0,1]}\|Z^n_t-Z_t\|^2\right)= 0.\]
{More specifically, 
\[\sqrt{\E\left(\sup_{t\in[0,1]}\|Z^n_t-Z_t\|^2\right)}\leq c \sqrt{d\over n}\]
holds for some constant $c<\infty$ which depends only on $\alpha, A, \lambda$.
Here $d$ denotes the number of edges $\nu=(i,j)$ in the network.}
\end{theorem}

The previous theorem shows that learning in BNNs based on the local principle of STDP may indeed lead to optimization of parameters according to SGD if the number of spikes per learning opportunity is high. To wit, we start with an initial parameter vector $\thf^{(0)}$
and loss function $L=L(\cdot,\X_1,\Y_1)$.
According to Theorem \ref{thm:convergence} the STDP nearly performs a continuous gradient step as in \eqref{eq:Z_stetig},
leading to an updated parameter vector $\thf^{(1)}$.
Switching now to the loss function $L=L(\cdot,\X_2,\Y_2)$, the next approximate gradient step leads to an updated vector $\thf^{(2)}$ etc.

This procedure only differs from the classical SGD in that we make many small instead of  
one large gradient step per learning opportunity.
Interestingly, neither the gradient nor even the functional dependence of the loss function on the parameters need to be known explicitly for this purpose.
By contrast, it relies crucially on the randomness in the update, which may seem counterintuitive because the desired gradient ODE \eqref{eq:Z_stetig} is deterministic.

\begin{proof}[Proof of Theorem \ref{thm:convergence}]
Following the argument in \cite[Proof of Theorem 1]{schmidt2023interpreting}, we may decompose the dynamics of $\Zf^n$ in coordinate $\nu=(i,j)$ as
\begin{align}\label{eq:Z_diff}
\Z^n_{t^n_{\ell},\nu}
&=\Z^n_{t^n_{\ell-1},\nu}+c_{\nu}^n\bigl(\Zf^n_{t^n_{\ell-1}}\bigr)+\D^n_{\ell,\nu},
\end{align}
where
\begin{align*}
b_\nu^n(z)&:=-n\alpha e^{-A_n} {C(A_n)\over 2A_n}\partial_\nu L(z),\\
c_{\nu}^n(z)&:=\frac{1}{n}\E b_\nu^n(z+\n\Vf^{\nu}),
\end{align*}
and
\begin{align*}
\D^n_{\ell,\nu} &:= \alpha \left(L\bigl(\Zf^n_{t^n_{\ell}}+\n\uf^{(\ell)}\bigr) - L\bigl(\Zf^n_{t^n_{\ell-1}}+\n\uf^{(\ell-1)}\bigr)\right)
\bigl(e^{-\n U^{(\ell)}_\nu}-e^{\n U^{(\ell)}_\nu} \bigr) - c_{\nu}^n\bigl(\Zf^n_{t^n_{\ell-1}}\bigr)
\end{align*}
is a martingale difference process with respect to the filtration $(\F_\ell)_{\ell=0,\dots,n}$ that is generated by all randomness up to step $\ell$.
Moreover, the random vector $\n\Vf^{\nu}$ has independent components where all but the $\nu$th are uniformly distributed 
on $[-A_n,A_n]$ and the $\nu$th has density 
\[f_{A_n}(x) := C(A_n)^{-1}(e^{A_n} - e^x)(e^{A_n} - e^{-x})\] 
on $[-A_n,A_n]$ 
with normalizing constant
\[C(A_n) := {\int_{-A_n}^{A_n} (e^{A_n} - e^x)(e^{A_n} - e^{-x})dx}={2A_n(e^{2A_n} + 1) + 2 - 2e^{2A_n}}.\]
{A Taylor expansion yields that
$C(A_n)/(2A_n)=2A_n^2/3+O(A_n^3)$ as $n\to\infty$.}
We obtain
\begin{align*}
\Zf^n_t-\Zf_t=\int_0^{\tn}\bigl(b_\nu(\Zf^n_{s})-b_\nu(\Zf_{s})\bigr)ds+\Cf^n_{\tn}+\Mf^n_{\tn}+\delf^n_t,
\end{align*}
 for all $t\in[0,1]$ where 
\begin{align*}
b_\nu(z)&:=-{{2\over3}A^2}\alpha\partial_\nu L(z),\\
\Cf^n_{\tn}&:=\sum_{\ell:\,t_\ell^n\leq t}
\left(\mathbf{c}^n\bigl(\Zf^n_{t^n_{\ell-1}}\bigr)-\frac{1}{n}b_\nu\bigl(\Zf^n_{t^n_{\ell-1}}\bigr)\right)\\
&=\sum_{\ell:\,t_\ell^n\leq t}\left(\mathbf{c}^n\bigl(\Zf^n_{t^n_{\ell-1}}\bigr)-\frac{1}{n}b_\nu^n\bigl(\Zf^n_{t^n_{\ell-1}}\bigr)\right)
+\frac{1}{n}\sum_{\ell:\,t_\ell^n\leq t}\left(b_\nu^n\bigl(\Zf^n_{t^n_{\ell-1}}\bigr)-b_\nu\bigl(\Zf^n_{t^n_{\ell-1}}\bigr)\right),\\
\Mf^n_{\tn}&:=\sum_{\ell:\,t_\ell^n\leq t}\Df^n_{t^n_{\ell}},\\
\delf^n_t&:=\Zf_{\tn}-\Zf_t.
\end{align*}
Set
\[\varepsilon_n(t):=\sqrt{\E\sup_{s\leq t}\|\Zf^n_s-\Zf_s\|^2}.\]
Using \eqref{ass:1} and the triangle inequality, we conclude that
\begin{align*}
\varepsilon_n(t)\leq \gamma_n+{{2\over3}A^2}\alpha\lambda\int_0^t\varepsilon_n(s)ds
\end{align*}
with
\begin{align*}
\gamma_n=\sqrt{\E\sup_{\ell=0,\dots,n}\bigl\|\Cf^n_{t^{(\ell)}_n}\bigr\|^2}+\sqrt{\E\sup_{\ell=0,\dots,n}\bigl\|\Mf^n_{t^{(\ell)}_n}\bigr\|^2}+{\sup_{t\in[0,1]}\|\delf^n_{t}\|}.
\end{align*}
By Grönwall's inequality we obtain
\[\varepsilon_n(t)\leq \gamma_n\exp\Bigl({{2\over3}A^2\alpha}\lambda\Bigr).\]
It is therefore sufficient to prove that {$\gamma_n\leq c \sqrt{d/n}$ for some constant $c$}. 

Note that
\[\left\|b_\nu^n\bigl(\Zf^n_{t^n_{\ell-1}}\bigr)-b_\nu\bigl(\Zf^n_{t^n_{\ell-1}}\bigr)\right\|\leq 
\alpha \biggl|{e^{-A_n}n{C(A_n)\over2A_n}-{2\over3}A^2}\biggr|\sup_z\|\partial_\nu L(z)\|\]
and 
\[\left\|\mathbf{c}^n\bigl(\Zf^n_{t^n_{\ell-1}}\bigr)-\frac{1}{n}b_\nu^n\bigl(\Zf^n_{t^n_{\ell-1}}\bigr)\right\|
\leq {1\over n}\alpha e^{-A_n}n{C(A_n)\over2A_n}\lambda A_n\]
because the random vectors $\n\Vf^{\nu}$ are all concentrated on $[-A_n,A_n]$.
Since {
\[e^{-A_n}n{C(A_n)\over2A_n}-{2\over3}A^2=O(n^{-1/2})\] 
and $A_n=O(n^{-1/2})$ as $n\to\infty$,
we obtain that 
$\E\sup_{\ell=0,\dots,n}\|\Cf^n_{t^{(\ell)}_n}\|^2=O(n^{-1/2})$}
as desired.
Moreover, we have that {$\sup_{t\in[0,1]}\|\delf^n_{t}\|=O(n^{-1})$} because
$\Zf$ is {Lipschitz}.

So it only remains to be verified that 
{$\E\sup_{\ell=0,\dots,n}\|\Mf^n_{t^{(\ell)}_n}\|^2\leq cd/n$ for some constant $c$.}
Doob's inequality yields
\[\E\sup_{\ell=0,\dots,n}\bigl\|\Mf^n_{t^{(\ell)}_n}\bigr\|^2\leq 4 \E\left(\sum_{\nu} \sum_{\ell=1}^n\left(\D^n_{{\ell},\nu}\right)^2\right).\]
Since the gradient of $L$ is bounded and the components of the $\n\uf^{(\ell)}$ are concentrated on the interval 
{$[-A_n,A_n]=[-n^{-1/2}A,n^{-1/2}A]$},
we have that 
\begin{equation}\label{e:doobterm}
\left\|\alpha \left(L\bigl(\Zf^n_{t^n_{\ell}}+\n\uf^{(\ell)}\bigr) - L\bigl(\Zf^n_{t^n_{\ell-1}}+\n\uf^{(\ell-1)}\bigr)\right)
\left(e^{-\n\uf^{(\ell)}}-e^{\n\uf^{(\ell)}} \right)\right\| \leq 
a{n^{-1}} + b \left\|\Zf^n_{t^n_{\ell}}-\Zf^n_{t^n_{\ell-1}}\right\|
\end{equation}
for some constants $a,b\in\mathbb R_+$.
By Lemma \ref{l:gronwall} below and $\Zf^n_{t^n_{0}}-\Zf^n_{t^n_{-1}}=0$ this implies that \eqref{e:doobterm}
is bounded by a multiple of {$n^{-1}$} and hence its square by {$n^{-2}$}.
Moreover, 
$\|c_{\nu}^n(\Zf^n_{t^n_{\ell-1}})\|$
is bounded by a multiple of {$n^{-1}$}.
Together, this yields that
$\sum_{\ell=1}^n(\D^n_{{\ell},\nu})^2$
is bounded by a multiple of {$n^{-1}$}, which yields the desired upper bound.
\end{proof}

\begin{lemma}\label{l:gronwall}
 $x_n\leq a+bx_{n-1}$, $n=1,2,\dots$ for $a,b,x_n\in\mathbb R_+$ implies
\[x_n\leq x_0 b^n+a{1-b^n\over 1-b}.\]
\end{lemma}
\begin{proof}
This follows by induction on $n$.
\end{proof}


\color{black}

\bibliography{lit}

\providecommand{\bysame}{\leavevmode\hbox to3em{\hrulefill}\thinspace}
\providecommand{\MR}{\relax\ifhmode\unskip\space\fi MR }
\providecommand{\MRhref}[2]{%
  \href{http://www.ams.org/mathscinet-getitem?mr=#1}{#2}
}
\providecommand{\href}[2]{#2}
\begin{thebibliography}{10}

\bibitem{conn2009introduction}
Andrew~R Conn, Katya Scheinberg, and Luis~N Vicente, \emph{Introduction to
  derivative-free optimization}, SIAM, 2009.

\bibitem{crick1989recent}
Francis Crick, \emph{The recent excitement about neural networks}, Nature
  \textbf{337} (1989), no.~6203, 129--132.

\bibitem{duchi2015optimal}
John~C Duchi, Michael~I Jordan, Martin~J Wainwright, and Andre Wibisono,
  \emph{Optimal rates for zero-order convex optimization: The power of two
  function evaluations}, IEEE Transactions on Information Theory \textbf{61}
  (2015), no.~5, 2788--2806.

\bibitem{grossberg1987competitive}
Stephen Grossberg, \emph{Competitive learning: From interactive activation to
  adaptive resonance}, Cognitive science \textbf{11} (1987), no.~1, 23--63.

\bibitem{hebb2005organization}
Donald~Olding Hebb, \emph{The organization of behavior: A neuropsychological
  theory}, Psychology press, 2005.

\bibitem{lillicrap2020backpropagation}
Timothy~P Lillicrap, Adam Santoro, Luke Marris, Colin~J Akerman, and Geoffrey
  Hinton, \emph{Backpropagation and the brain}, Nature Reviews Neuroscience
  \textbf{21} (2020), no.~6, 335--346.

\bibitem{liu2021single}
Xiao Liu, David~A Leopold, and Yifan Yang, \emph{Single-neuron firing cascades
  underlie global spontaneous brain events}, Proceedings of the National
  Academy of Sciences \textbf{118} (2021), no.~47, e2105395118.

\bibitem{nesterov2017random}
Yurii Nesterov and Vladimir Spokoiny, \emph{Random gradient-free minimization
  of convex functions}, Foundations of Computational Mathematics \textbf{17}
  (2017), 527--566.

\bibitem{orvieto2019shadowing}
Antonio Orvieto and Aurelien Lucchi, \emph{Shadowing properties of optimization
  algorithms}, Advances in Neural Information Processing Systems \textbf{32}
  (2019).

\bibitem{schmidt2023interpreting}
Johannes Schmidt-Hieber, \emph{Interpreting learning in biological neural
  networks as zero-order optimization method}, arXiv preprint arXiv:2301.11777
  (2023).

\bibitem{tavanaei2019deep}
Amirhossein Tavanaei, Masoud Ghodrati, Saeed~Reza Kheradpisheh, Timoth{\'e}e
  Masquelier, and Anthony Maida, \emph{Deep learning in spiking neural
  networks}, Neural networks \textbf{111} (2019), 47--63.

\bibitem{whittington2019theories}
James~CR Whittington and Rafal Bogacz, \emph{Theories of error back-propagation
  in the brain}, Trends in cognitive sciences \textbf{23} (2019), no.~3,
  235--250.

\end{thebibliography}
\bibliographystyle{amsplain}
\end{document}